\documentclass[11pt]{article}
\usepackage[T1]{fontenc}
\usepackage{lmodern}
\usepackage{amsmath,amssymb,amsthm}
\usepackage{booktabs,tabularx}
\usepackage{flafter}
\usepackage[margin=1in]{geometry}
\usepackage[colorlinks=true,allcolors=blue]{hyperref}
\newtheorem{theorem}{Theorem}
\newtheorem{lemma}[theorem]{Lemma}
\newtheorem{corollary}[theorem]{Corollary}
\newtheorem{definition}[theorem]{Definition}
\newcommand{\eps}{\varepsilon}

\title{Quadratic Word Equations with a Linear Side: Polynomial Nielsen Graph Diameter and NP-Completeness}
\author{Yuki Yonemoto\thanks{This work was supported by JST BOOST, Japan Grant Number JPMJBS2406.}\\
Kyushu University, Fukuoka, Japan\\
\texttt{yonemoto.yuuki.240@s.kyushu-u.ac.jp}\\
\href{https://orcid.org/0009-0008-5330-7256}{ORCID: 0009-0008-5330-7256}}
\date{October 2026}

\usepackage{tikz}
\usetikzlibrary{arrows.meta,positioning,fit,backgrounds}

\hypersetup{
 pdftitle={Quadratic Word Equations with a Linear Side: Polynomial Nielsen Graph Diameter and NP-Completeness},
 pdfauthor={Yuki Yonemoto},
 pdfkeywords={Word equations, quadratic word equations, Nielsen transformations, NP-completeness}
}
\begin{document}
\maketitle
\begin{abstract}
The satisfiability problem for word equations asks whether variables can be
replaced by words so that the two sides become equal. For regular word
equations, in which each variable occurs at most once on each side,
satisfiability is NP-complete. For general quadratic word equations, in
which each variable occurs at most twice in total, satisfiability is
NP-hard, but its membership in NP remains open. We consider an intermediate
class: quadratic word equations with a linear side, where each variable
occurs at most once on one designated side. We show that the Nielsen graph
of an equation $U=V$ in this class, with total length $N=|U|+|V|$, has
diameter $O(N^{12})$, measured over reachable pairs of vertices.
Together with the known NP-hardness for regular word equations, this result
establishes NP-completeness of satisfiability for this class.
We also show that each strongly connected component is isomorphic to the
length-preserving reachability graph of a regular equation, and that the
condensation graph has depth at most $|U|+2|V|$.
For arbitrary quadratic word equations, we further show that Nielsen
paths can be shortened to a length polynomial in the equation size and
the number of prefix operations between variables repeated on opposite
sides. This yields NP membership for subclasses defined by variable
occurrence arrangements.
\end{abstract}

\section{Introduction}\label{sec:introduction}

A word equation is an equality between two words consisting of constants
and variables. Its satisfiability problem asks whether the variables can
be replaced by words of constants so that the two sides become equal.
Word equations provide a basic framework for reasoning about concatenation,
and arise in the study of formal languages, unification, and string
constraints in program verification~\cite{hardness_simple,lin_majumdar}.
Throughout this paper, variable images may be empty, and we consider single
equations without additional constraints on their solutions.

Makanin established the decidability of word equations in
1977~\cite{makanin}. Subsequent work has substantially improved the known
complexity bounds. Plandowski proved that satisfiability is in
PSPACE~\cite{plandowski}, and Je\.{z} later obtained a nondeterministic
linear-space algorithm~\cite{jez_linear}.
Nevertheless, whether satisfiability belongs to NP remains open, even for
quadratic word equations~\cite{day_manea}. This motivates the study of
restrictions on the occurrences of variables.

\paragraph{Occurrence restrictions and known bounds.}
An equation is \emph{quadratic} if each variable occurs at most twice in
total.
An equation is \emph{regular} if each variable occurs at most once on each
side. Thus every regular equation is quadratic, but a quadratic equation
may contain two occurrences of a variable on the same side.
Day, Manea, and Nowotka showed NP-hardness even for regular-ordered
equations, in which variables occurring on both sides appear in the
same order on each side of a regular equation~\cite{hardness_simple}.
Day and Manea subsequently established NP membership for all regular
equations, and hence NP-completeness, by bounding the diameter of their
Nielsen graphs~\cite{day_manea}.
For unrestricted quadratic equations, NP-hardness is known but NP
membership remains unresolved. This question has been open at least
since the 1999 work of Robson and Diekert~\cite{robson_diekert_1999},
who established NP-hardness even for a single quadratic equation and
gave a linear-time decision procedure when the lengths of the variable
images are specified.
Table~\ref{tab:complexity} summarizes these bounds and the result
of this paper.

\begin{table}[t]
\centering
\small
\renewcommand{\arraystretch}{1.2}
\begin{tabularx}{\textwidth}{@{}p{0.16\textwidth}X p{0.23\textwidth}p{0.14\textwidth}@{}}
\toprule
Class & Variable occurrences & Satisfiability & Reference \\
\midrule
Regular & At most once on each side & NP-complete & \cite{hardness_simple,day_manea} \\
\textbf{Quadratic with a linear side} & At most twice in total; at most once on one side & \textbf{NP-complete} & \textbf{This work} \\
Quadratic & At most twice in total & NP-hard; in PSPACE & \cite{hardness_simple,plandowski} \\
General & Unrestricted & NP-hard; in PSPACE & \cite{hardness_simple,plandowski} \\
\bottomrule
\end{tabularx}
\caption{Complexity of satisfiability for single word equations with empty
variable images allowed and no additional constraints. NP membership is
open for the unrestricted quadratic and general classes.
The PSPACE upper bounds follow from the general case~\cite{plandowski}.
These classes also admit the nondeterministic linear-space upper
bound of Je\.{z}~\cite{jez_linear}, measured in the encoded input length.
The lower bounds for the larger classes follow by inclusion.}
\label{tab:complexity}
\end{table}

\paragraph{An intermediate class.}
We consider quadratic equations with one linear side. More precisely, an
input is an equation $U=V$ such that each variable occurs at most twice in
$UV$ and at most once in $V$, after exchanging the sides if necessary.
The question is whether the equation has a solution.
This condition relaxes regularity by allowing variables to occur twice on
one side, while retaining linearity on the other. It places no bound on
the number of variables repeated on the nonlinear side. The resulting
class lies strictly between regular and quadratic equations: for distinct
variables $x,y$, the equation $xx=y$ belongs to our class but is not regular,
whereas $xx=yy$ is quadratic and has no linear side.

\paragraph{Polynomial diameter and NP-completeness.}\label{sec:main}
Nielsen transformations turn satisfiability of quadratic equations into
reachability of the empty equation $(\eps,\eps)$. The Nielsen graph has
equations as vertices and individual transformations as edges. A path is
accepting if it ends at $(\eps,\eps)$. We count each elementary prefix
transformation as one step; the precise rules are given in
Section~\ref{subsec:nielsen}. For a reachable pair $E,F$, let $d(E,F)$ be
the minimum number of steps from $E$ to $F$.
Our main result bounds this distance for every target reachable from an
initial equation with a linear side.

\begin{theorem}\label{thm:main}
Let $E=(U,V)$ be a quadratic word equation with a linear right side,
and let $N=|E|$. For every equation $F$ reachable from $E$,
$d(E,F)=O(N^{12})$.
\end{theorem}
Since quadraticity and right linearity are preserved, and the total length
does not increase, the theorem applies at every reachable equation. Thus the
Nielsen graph reachable from $E$ has diameter $O(N^{12})$, where the
diameter is the maximum finite directed distance over reachable pairs.
In particular, every satisfiable equation in this class admits a
polynomial-length accepting path.

\begin{corollary}\label{cor:np}
Satisfiability of quadratic word equations with one linear side is
NP-complete.
\end{corollary}
Indeed, an accepting path is a polynomial-size certificate: it has
polynomially many equations, each of at most $N$ symbols, and each transition
can be checked in polynomial time. Soundness and completeness of the rules
are recalled in Lemma~\ref{lem:acceptance}.
NP-hardness is inherited from regular equations~\cite{hardness_simple},
which form a subclass. Thus the new upper bound supplies NP membership;
the lower bound is already known for a smaller class.

\paragraph{Why this works although the equation is not regular.}
The main difficulty is that an equation in our class may contain
linearly many variables occurring twice on the nonlinear side, so it
need not be close to regular in terms of variable occurrences.
Nevertheless, the effect of this nonregularity on a Nielsen path can
be confined to only $O(N)$ exceptional steps, where $N$ is the input
size. More precisely, call a variable shared if it occurs once on each
side. A nonnegative potential bounds the number of operations other
than prefix extensions of shared variables by $2N$.
Between these exceptional steps, the number of occurrences of each
symbol on each side remains fixed. Replacing the variables occurring
twice on the nonlinear side by distinct fresh constants therefore gives
an isomorphism between the graph of shared extensions and the
length-preserving reachability graph of a regular equation.
This replacement is chosen separately for each segment; it is a
correspondence between transformation graphs, not a reduction preserving
solution sets.
The regular diameter bound of Day and Manea~\cite{day_manea}
then shortens each segment while preserving its endpoints.
Figure~\ref{fig:proof-overview} illustrates how these segments and
exceptional steps fit together; Section~\ref{sec:proof-outline} gives
the formal proof outline.

\paragraph{Structure of the Nielsen graph.}
Beyond the path-length bound, we describe the strongly connected
components (SCCs) of the full Nielsen graph. Shared extensions are
reversible by finite rotations, and their reachability classes are
exactly the SCCs. For every reachable equation, replacing the variables
occurring twice on its left side by distinct fresh constants identifies
its SCC with the length-preserving reachability graph of a regular
equation. This directed graph isomorphism preserves all distances within
the component, as well as its diameter, number of vertices, and directed
cycle structure. The potential is constant on each SCC and strictly
decreases between components, so the condensation graph has depth at
most $|U|+2|V|\le2N$ (Theorem~\ref{thm:scc-decomposition} and
Corollary~\ref{cor:scc-distance}). Thus the linear-depth decomposition
known for regular equations~\cite[Section~3.4]{day_manea} extends to
quadratic equations with a linear side, with each component represented
exactly by the length-preserving reachability graph of a regular equation.

\paragraph{Extension to quadratic equations.}
For arbitrary quadratic word equations of size $N$ with $m$ variables,
we show that any Nielsen path containing $r$ prefix operations between
variables repeated on opposite sides can be shortened to length
$O((m+r+N)N^{11})$ (Theorem~\ref{thm:R-shortening}). In particular, if an accepting path never repeats the same ordered
pair of extended and prefix variables in these operations within a
length-preserving segment, it can be shortened to length $O(N^{14})$. We use this criterion to establish NP membership for the aligned-block
and cyclic-bridge classes.

\paragraph{Related work.}
Beyond regular equations, Day, Manea, and Nowotka established NP membership
for variable-sparse quadratic equations and for $k$-ordered quadratic
equations for every fixed $k$~\cite{quadratic_minimal_solutions}.
The former bound the number of variables occurring twice in total,
including variables shared by the two sides, by the logarithm of the
equation length. The latter admit a decomposition of both sides into
blocks over pairwise disjoint variable sets appearing in a common order,
with at most $k$ variables per set. Our class imposes neither a bound on
the number of repeated variables nor such a block-order restriction.
Thus these results do not directly cover all quadratic equations with a
linear side. The occurrence-based subclasses studied here are likewise
not contained in the variable-sparse or any fixed-$k$-ordered class:
they allow linearly many repeated variables and may require arbitrarily
large variable blocks in a $k$-ordered decomposition.

Another direction adds constraints on the words assigned to variables.
Lin and Majumdar study quadratic equations with length constraints via
counter systems and arithmetic with divisibility, including decision
procedures for restricted classes~\cite{lin_majumdar}. These are different
problems from the unconstrained satisfiability problem considered here;
our bounds concern paths in the equation graph without such constraints.

\paragraph{Organization.}
Section~\ref{sec:preliminaries} fixes the notation and transformation rules.
Section~\ref{sec:bound} presents the proof outline, derives the main theorem
from three lemmas, and then proves the lemmas in turn.
Section~\ref{sec:scc-decomposition} establishes the SCC decomposition
and distance preservation.
Section~\ref{sec:extensions} extends the path bound to quadratic word
equations and applies it to several subclasses.
Section~\ref{sec:scope} discusses further directions.

\section{Preliminaries}\label{sec:preliminaries}

\subsection{Words}
Let $\Sigma$ be a finite alphabet. A \emph{word} over $\Sigma$ is a finite
sequence of symbols from $\Sigma$, and $\Sigma^*$ denotes the set of
all such words. For a word $w$, its length $|w|$ is the number of
symbols in $w$. We write $\eps$ for the empty word, so $|\eps|=0$.
Concatenation of words $u$ and $v$ is denoted by $uv$.
For a symbol $s$, let $|w|_s$ denote its number of occurrences in $w$.

\subsection{Word equations}
Let $X$ be a set of variables disjoint from $\Sigma$. Symbols in
$\Sigma$ are called \emph{constants}. A \emph{word equation} is a pair
$E=(U,V)$ of words over $\Sigma\cup X$, also written $U=V$. We put
\[
 |E|=|U|+|V|,
 \qquad
 \operatorname{Var}(E)=\{x\in X:|U|_x+|V|_x>0\}.
\]
Thus the size counts occurrences of symbols, rather than the lengths of
the words substituted for variables.
An assignment $h:X\to\Sigma^*$ maps each variable to a word of constants.
We extend it to constants by $h(a)=a$ for $a\in\Sigma$, and to words by
\[
 h(s_1\cdots s_m)=h(s_1)\cdots h(s_m),\qquad h(\eps)=\eps.
\]
An assignment $h$ is a \emph{solution} of $E$ if $h(U)=h(V)$, and $E$ is
\emph{satisfiable} if it has a solution. Variable images may be empty;
we impose no additional constraints on their lengths or contents.

An equation $(U,V)$ is \emph{quadratic} if $|U|_x+|V|_x\le2$ for every
$x\in X$. A word is \emph{linear} if each variable occurs in it at most
once. An equation is \emph{regular} if both its sides are linear.
A variable of $(U,V)$ is \emph{shared} if $|U|_x=|V|_x=1$.

The main class studied in this paper consists of quadratic equations with
one linear side. We orient such an equation initially so that its right
side is linear, and keep that orientation throughout the argument.
For example, $xax=y$ has a linear right side and is quadratic, but is not
regular, where $a$ is a constant and $x,y$ are distinct variables.
Linearity places no restriction on occurrences of constants.

\subsection{Nielsen transformations}\label{subsec:nielsen}
We use elementary prefix Nielsen transformations, as in the formulation
of Lin and Majumdar~\cite[Section~3.1]{lin_majumdar}, with the empty-side
case specified explicitly below. All substitutions act simultaneously on
every occurrence in both sides. Symbols not mentioned in a substitution
are fixed. We write $E\to F$ when one of the following rules transforms
$E$ into $F$.

\begin{enumerate}
\item \textbf{Erasure.}
If a variable $x$ is the first symbol of either side, apply the substitution
$x\mapsto\eps$ to both sides.
\item \textbf{Cancellation.}
For any symbol $s\in\Sigma\cup X$, replace $(sU,sV)$ by $(U,V)$.
\item \textbf{Prefix extension.}
Suppose that both sides are nonempty and have distinct first symbols,
one a variable $x$ and the other a symbol $y$.
Let $\theta$ be the substitution $x\mapsto yx$. Apply $\theta$ to both
sides and then cancel one leading $y$ from each. In particular,
\[
 (xU,yV)\ \to\ (x\theta(U),\theta(V)),
\]
with the symmetric rule when $x$ is the first symbol of the right side.
We call this an \emph{extension of $x$}. The extended variable is $x$,
and $y$ is the prefix symbol, or the prefix variable when it is a variable.
\end{enumerate}
If the distinct leading symbols are both variables, either extension is
available. If they coincide and are a variable, both erasure and
cancellation are available. With exactly one empty side, only erasure of
the leading variable of the other side is available. Distinct leading
constants admit no rule. The equation $(\eps,\eps)$ has no outgoing edge.
These are syntactic rules; no nonemptiness assumptions on variable images
are recorded in an equation.

\subsection{Paths and acceptance}\label{subsec:paths}
For an equation $E$, its \emph{Nielsen graph} $G(E)$ has as vertices all
equations reachable from $E$ by these rules, and an edge for each rule
application. Vertices retain their symbol names and the order of their
sides. In particular, we do not take a quotient by variable renaming or
exchange of sides.
A path is a finite sequence of successive rule applications, with repeated
vertices allowed, and its length is its number of edges. Each rule counts
as one edge, including the cancellation incorporated into a prefix
extension. No further cancellation is performed for free.
We write $E\to^*F$ for reachability, allowing a path of length zero, and
$d(E,F)$ for the shortest path length when $F$ is reachable from $E$.
A path is \emph{accepting} if its last equation is $(\eps,\eps)$.
For a finite directed graph, we use \emph{diameter} to mean the maximum
shortest-path distance over ordered pairs of vertices for which a directed
path exists; unreachable pairs are excluded.
An edge is \emph{length-preserving} if its endpoints have the same size;
the length-preserving reachability graph is obtained by allowing only
such edges from the specified initial equation.

\begin{lemma}[{\cite[Section~3.1]{lin_majumdar}}]\label{lem:quadratic-size}
If $E$ is quadratic and $E\to F$, then $F$ is quadratic and $|F|\le|E|$.
Consequently, $G(E)$ is finite for every quadratic equation $E$.
\end{lemma}

We recall the standard characterization of satisfiability by acceptance.
\begin{lemma}[{\cite[Proposition~1]{lin_majumdar}}]\label{lem:acceptance}
A quadratic word equation $E$ is satisfiable if and only if
$E\to^*(\eps,\eps)$.
\end{lemma}
The empty-side convention above makes explicit the terminal case of the
procedure: an equation with one empty side is satisfiable exactly when
the other side contains only variables, which can be erased successively.

\section{Proof of the main theorem}\label{sec:bound}

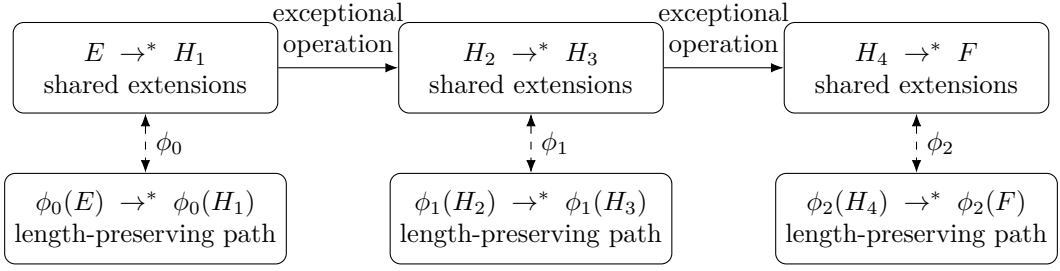
\begin{figure}[t]
\centering
\begin{tikzpicture}[
  >=Latex,
  block/.style={draw,rounded corners,minimum width=35mm,minimum height=12mm,align=center,font=\small},
  lab/.style={font=\small,align=center}
]
\node[block] (S0) at (0,0) {$E\;\to^*\;H_1$\\shared extensions};
\node[block] (S1) at (5.1,0) {$H_2\;\to^*\;H_3$\\shared extensions};
\node[block] (S2) at (10.2,0) {$H_4\;\to^*\;F$\\shared extensions};
\draw[->] (S0) -- node[above,lab] {exceptional\\operation} (S1);
\draw[->] (S1) -- node[above,lab] {exceptional\\operation} (S2);
\node[block] (R0) at (0,-2) {$\phi_0(E)\;\to^*\;\phi_0(H_1)$\\length-preserving path};
\node[block] (R1) at (5.1,-2) {$\phi_1(H_2)\;\to^*\;\phi_1(H_3)$\\length-preserving path};
\node[block] (R2) at (10.2,-2) {$\phi_2(H_4)\;\to^*\;\phi_2(F)$\\length-preserving path};
\draw[<->,dashed] (S0) -- node[right,lab] {$\phi_0$} (R0);
\draw[<->,dashed] (S1) -- node[right,lab] {$\phi_1$} (R1);
\draw[<->,dashed] (S2) -- node[right,lab] {$\phi_2$} (R2);
\end{tikzpicture}
\caption{A path with two exceptional operations and three segments of
shared extensions. In general, $r\le2N$ exceptional operations separate
$r+1$ segments, which may be empty. Each segment has its own replacement
map $\phi_i$ and corresponds to a length-preserving path for a regular
equation. Shortening those paths and mapping them back preserves the
segment endpoints, giving a total length of $O(N^{12})$.}
\label{fig:proof-overview}
\end{figure}

\subsection{Proof outline}\label{sec:proof-outline}
Fix an initial quadratic equation $E=(U,V)$ with linear right side, and
write $N=|E|$. We first establish that right linearity is preserved, so
that the same occurrence restrictions apply throughout a Nielsen path.

\begin{lemma}\label{lem:right-linear}
If $H$ is quadratic with linear right side and $H\to H'$, then $H'$ is
quadratic with linear right side and $|H'|\le |H|$.
\end{lemma}
We prove Lemma~\ref{lem:right-linear} in Subsection~\ref{sec:preservation}.

Recall that a variable is \emph{shared} if it occurs once on each side.
We divide the elementary operations into two types: prefix extensions of
shared variables, called \emph{shared extensions}, and all remaining
operations, which we call \emph{exceptional operations}. The latter include prefix extensions of nonshared variables,
erasures, and cancellations. Lemma~\ref{lem:right-linear} ensures that
this classification applies to every reachable equation.

The next two lemmas control the two types separately. First, we bound
the total number of operations other than shared extensions.
\begin{lemma}\label{lem:potential}
Any path starting at $E$ contains at most $2N$ operations other than
shared extensions.
\end{lemma}
We prove Lemma~\ref{lem:potential} in Subsection~\ref{sec:potential}.

Shared extensions need not be few, but a segment consisting only of such
operations can be replaced by a short segment with the same endpoints.
\begin{lemma}\label{lem:shared-shortening}
For every quadratic equation $H$ with
linear right side and every $K$ reachable from $H$ using only shared
extensions, there is a path of shared extensions from $H$ to $K$ of
length at most $C_0|H|^{11}$, for a constant $C_0$ independent of $H$
and $K$.
\end{lemma}
We prove Lemma~\ref{lem:shared-shortening} in Subsection~\ref{sec:freezing},
by replacing the variables occurring twice on the left by pairwise
distinct fresh constants. Lemma~\ref{lem:freezing} shows that paths
correspond under this replacement, and Lemma~\ref{lem:regular-distance}
provides the length bound for the resulting regular equation.
Figure~\ref{fig:proof-overview} illustrates this decomposition and the
replacement of each segment. Before proving the lemmas, we show how
they yield the main theorem.

\begin{proof}[Proof of Theorem~\ref{thm:main}]
Choose any path from $E$ to a reachable equation $F$.
Let $r$ be the number of operations on this path that are not shared
extensions. By Lemma~\ref{lem:potential}, $r\le 2N$.
Splitting the path at these operations gives $r+1$ intervening segments
of shared extensions; these segments may have length $0$.
By Lemma~\ref{lem:right-linear}, the initial equation of every segment is
quadratic with linear right side and has size at most $N$.
Lemma~\ref{lem:shared-shortening} therefore replaces each segment by one
of length at most $C_0N^{11}$ with the same endpoints.
The operations between segments remain valid because their endpoints
have not changed. Concatenating the replacement segments and these
operations gives a path from $E$ to $F$ of length at most
\[
 (r+1)C_0N^{11}+r
 \le (2N+1)C_0N^{11}+2N
 =O(N^{12}).
\]
This proves the required bound.
\end{proof}

\subsection{Preservation of right linearity}\label{sec:preservation}
We prove Lemma~\ref{lem:right-linear} by checking that each elementary operation preserves right linearity.

\begin{proof}[Proof of Lemma~\ref{lem:right-linear}]
Quadraticity and nonincrease of total size follow from
Lemma~\ref{lem:quadratic-size}. It remains to check right linearity.
Erasure and cancellation only remove symbols, so they preserve this
property. For a variable $x$ and a symbol $y$, consider a prefix 
extension $x\mapsto yx$, where $x\ne y$,
and let $c_R(x)$ and $c_R(y)$ denote the numbers of occurrences of $x$
and $y$ on the right, respectively.
Under this operation, $c_R(x)$ is unchanged, while $c_R(y)$ changes by
$c_R(x)-1$: the
substitution inserts $c_R(x)$ copies, and the ensuing cancellation removes
one. If $y$ is a variable, $c_R(y)$ does not increase, since $c_R(x)\le 1$. 
Every other variable keeps its number of occurrences on the right. 
Thus the right side remains linear. 
\end{proof}

\subsection{Bounding the remaining operations}\label{sec:potential}
We prove Lemma~\ref{lem:potential} by using a nonnegative potential to bound the number of operations other than shared extensions.

\begin{proof}[Proof of Lemma~\ref{lem:potential}]
For an equation $H=(L,R)$ with $E\to^*H$, define the potential
$P(H)=|L|+2|R|$.
For each variable $x$, let $c_L(x)$ and $c_R(x)$ denote its numbers of
occurrences in $L$ and $R$, respectively.
We show that shared extensions leave $P$ unchanged, whereas every other
operation decreases $P$ by at least one.

First, consider a prefix extension $x\mapsto yx$, where $x$ is a variable and
$y$ is a symbol distinct from $x$. The substitution inserts
one copy of $y$ before each occurrence of $x$, and the subsequent cancellation
removes one leading $y$ from each side. Thus the changes in the side
lengths are
\[
 (\Delta|L|,\Delta|R|)=(c_L(x)-1,c_R(x)-1).
\]
By Lemma~\ref{lem:right-linear}, the possible pairs $(c_L(x),c_R(x))$ are
$(1,1),(2,0),(1,0),(0,1)$, with potential changes $0,-1,-2,-1$,
respectively. Here, the pair $(1,1)$ corresponds to a shared extension,
while each of the other pairs yields a strict decrease in $P$.

Cancellation decreases $P$ by $3$, and erasure of a variable $x$ decreases
$P$ by $c_L(x)+2c_R(x)\ge 1$.
Thus, shared extensions leave $P$ unchanged, whereas every other operation
decreases $P$ by at least one. Since $P$ is nonnegative, any path
starting at $E$ contains at most $P(E)=|U|+2|V|\le 2N$ operations
other than shared extensions.

\end{proof}

\subsection{Shortening shared-extension segments}\label{sec:freezing}
We prove Lemma~\ref{lem:shared-shortening} using two auxiliary lemmas:
a correspondence with a regular Nielsen graph and a distance bound for
that graph.

Let $H=(L,R)$ be a quadratic equation with a linear right side.
Let $B$ be the set of variables occurring twice in $L$.
For each $b\in B$, introduce a constant $\alpha_b$ not in
$\Sigma\cup X$, so that it does not occur in $H$. These symbols are
pairwise distinct.
Let $\phi$ replace each $b\in B$ by $\alpha_b$ and
leave every other symbol unchanged. Apply $\phi$ to words symbol by
symbol and to equations on both sides. The equation $\phi(H)$ is regular
and satisfies $|\phi(H)|=|H|$.

\begin{lemma}\label{lem:freezing}
For every equation $K$ reachable from $H$ using only shared extensions
and every integer $\ell\ge 0$, the following are equivalent:
\begin{enumerate}
\item[(1)] There is a path of shared extensions from $H$ to $K$ of length $\ell$.
\item[(2)] There is a path of length-preserving operations from $\phi(H)$
to $\phi(K)$ of length $\ell$.
\end{enumerate}
\end{lemma}

\begin{proof}[Proof of Lemma~\ref{lem:freezing}]
We first show that shared extensions preserve the number of occurrences
of each symbol on each side. Consider a shared extension $x\mapsto yx$,
where $x$ occurs once on each side and $y$ is a symbol distinct from $x$.
The number of occurrences of $x$ is unchanged. On each side, one copy
of $y$ is inserted before $x$ and one leading $y$ is then removed, so
its number of occurrences is also unchanged. Every other symbol is
unaffected. Thus the number of occurrences of each symbol on each side
is preserved, and $B$ remains fixed throughout a path of shared extensions.

We first prove $(1)\Rightarrow(2)$.
Let $H=H_0\to H_1\to\cdots\to H_\ell=K$ be a path of shared
extensions. For each operation $H_i\to H_{i+1}$, write the prefix
extension as $x\mapsto yx$. Since $x$ is shared, $x\notin B$ and
$\phi(x)=x$. Moreover, $\phi$ maps distinct symbols occurring in these
equations to distinct symbols, so $x\ne\phi(y)$ and the leading-symbol
condition is preserved. Substitution and cancellation therefore give
\[
 \phi(H_i)\to\phi(H_{i+1})
\]
by the operation $x\mapsto\phi(y)x$. The variable $x$ remains shared,
so this operation preserves length. Applying $\phi$ to the entire path
gives a length-preserving path of length $\ell$ from $\phi(H)$ to $\phi(K)$.

Next, we prove $(2)\Rightarrow(1)$.
Let $\phi(H)=Q_0\to Q_1\to\cdots\to Q_\ell=\phi(K)$ be a path
of length-preserving operations. Each $Q_i$ is regular: in a regular
equation, a length-preserving operation must extend a shared variable,
since cancellation, erasure, and extension of a variable occurring only
once all decrease the size. Such an extension preserves the number of
occurrences of each symbol on each side, by the argument above, and
hence preserves regularity.
For each $i$, let $H_i$ be obtained from $Q_i$ by replacing each
$\alpha_b$ by $b$. None of the variables in $B$ occurs in $Q_0$, and
the operations introduce no new symbols, so none occurs in any $Q_i$.
Thus this replacement preserves distinctness of symbols. A variable
extended in $Q_i$ is not one of the constants $\alpha_b$ and remains
shared in $H_i$. Replacing a leading $\alpha_b$ by $b$, if necessary,
therefore leaves the corresponding prefix extension applicable, and
substitution followed by cancellation produces $H_{i+1}$.
Consequently, $H_0\to\cdots\to H_\ell$ is a path of shared extensions.
Since $H_0=H$ and $H_\ell=K$, it has the required endpoints and length.
\end{proof}

We use the following distance bound from the proof of the regular
Nielsen graph diameter theorem of Day and Manea.
\begin{lemma}[{\cite[proof of Theorem~8.11]{day_manea}}]
\label{lem:regular-distance}
Let $Q$ be a regular word equation, and let $n=|Q|$.
For every equation $Q'$ reachable from $Q$ using only length-preserving
operations, there is a path of length-preserving operations from $Q$ to
$Q'$ of length at most $C_1n^{11}$, for a constant $C_1$ independent of
$Q$, $Q'$, and the constant alphabet.
\end{lemma}

This is the bound for length-preserving paths obtained within the proof
of \cite[Theorem~8.11]{day_manea}, rather than the exponent-$12$ bound
stated there for the full graph.
The final paragraph of that proof derives the $O(n^{11})$ bound for
length-preserving graphs using Theorem~4.8 and Remark~4.6 of the same paper.
The length-preserving operations in
\cite[Section~3.1]{day_manea} are the same prefix extensions of shared
variables used here, with the leading cancellation included in one step.
The bound is in terms of the equation size and applies to arbitrary
finite constant alphabets; in particular, it applies to the regular
equation $\phi(H)$ over the enlarged alphabet
$\Sigma\cup\{\alpha_b:b\in B\}$.

\begin{proof}[Proof of Lemma~\ref{lem:shared-shortening}]
Let $K$ be reachable from $H$ using only shared extensions.
By Lemma~\ref{lem:freezing}, $\phi(K)$ is reachable from $\phi(H)$
using only length-preserving operations. Since $\phi(H)$ is regular,
Lemma~\ref{lem:regular-distance} gives such a path of length at most
$C_1|\phi(H)|^{11}$. By the implication $(2)\Rightarrow(1)$ in
Lemma~\ref{lem:freezing}, there is a path of shared extensions from
$H$ to $K$ of the same length. As $|\phi(H)|=|H|$, its length is at most $C_1|H|^{11}$.
Taking $C_0=C_1$ proves the lemma.
\end{proof}

\section{Structure of the Nielsen graph}\label{sec:scc-decomposition}

The correspondence in Lemma~\ref{lem:freezing} also yields a structural
decomposition of the full Nielsen graph. We first record the reversibility of shared
extensions, using the rotation argument underlying
\cite[Remark~3.2]{day_manea} for regular equations.

\begin{lemma}\label{lem:shared-reversible}
If $H\to K$ is a shared extension, then $K$ reaches $H$ using only
shared extensions.
\end{lemma}
\begin{proof}[Proof of Lemma~\ref{lem:shared-reversible}]
Let $x$ be the extended variable. By symmetry, write
\[
 H=(xA,s_1\cdots s_m xB),\qquad m\ge1.
\]
Since $x$ is shared, none of $A$, $B$, and $s_1\cdots s_m$ contains $x$.
The extension $x\mapsto s_1x$ gives
\[
 K=(xA,s_2\cdots s_ms_1xB).
\]
Thus the left side stays fixed and the nonempty prefix before $x$ on
the right rotates one position to the left. Every symbol of this prefix
is distinct from $x$, so the same type of shared extension remains
applicable. After $m$ extensions the equation is again $H$; hence
$m-1$ further extensions suffice from $K$. If $m=1$, the edge is a
self-loop and no further operation is needed.
\end{proof}

Fix a quadratic equation $E$ with linear right side. For each
$H\in G(E)$, choose the replacement $\phi$ of
Subsection~\ref{sec:freezing} using the variables repeated on the left
side of $H$, and keep this choice fixed on its shared-extension graph.
Let $G_{\mathrm{sh}}(H)$ be the graph
obtained by allowing only shared extensions from $H$, and let
$G_{\mathrm{lp}}(\phi(H))$ be the length-preserving reachability graph
of $\phi(H)$. A \emph{strongly connected component} (SCC) is a maximal
set of mutually reachable vertices. The \emph{condensation graph}
is obtained by contracting each SCC to one vertex and discarding
all edges internal to a component. Its depth is the maximum number
of edges in a directed path.

\begin{theorem}[SCC decomposition]\label{thm:scc-decomposition}
For a quadratic equation $E=(U,V)$ with linear right side,
the following hold.
\begin{enumerate}
\item[(1)] \textbf{Identification of SCCs.}
The SCCs of $G(E)$ are exactly its shared-extension reachability classes.
\item[(2)] \textbf{Correspondence with regular graphs.}
For every $H\in G(E)$, the replacement $\phi$ from
Subsection~\ref{sec:freezing} induces an isomorphism of directed graphs
\[
 G_{\mathrm{sh}}(H)\cong G_{\mathrm{lp}}(\phi(H)).
\]
\item[(3)] \textbf{Condensation depth.}
The condensation graph of $G(E)$ has depth at most
$|U|+2|V|\le2|E|$.
\end{enumerate}
\end{theorem}
\begin{proof}[Proof of Theorem~\ref{thm:scc-decomposition}]
\noindent(1) Shared extensions are reversible by
Lemma~\ref{lem:shared-reversible}. Every other edge strictly decreases
$P(L,R)=|L|+2|R|$, which no edge increases
(proof of Lemma~\ref{lem:potential}), so it cannot lie within an SCC.
Hence the SCCs are exactly the shared-extension reachability classes.

\noindent(2) The proof of Lemma~\ref{lem:freezing} gives a bijection
on reachable vertices and edges, with inverse obtained by
replacing each fresh constant $\alpha_b$ by $b$. This is the claimed
isomorphism.

\noindent(3) By (1), $P$ is constant on each SCC and decreases by at
least one along each condensation edge. Since $0\le P\le P(E)$,
every directed path in the condensation has at most
$P(E)=|U|+2|V|$ edges.
\end{proof}

\begin{corollary}[Distance preservation]\label{cor:scc-distance}
Under the assumptions of Theorem~\ref{thm:scc-decomposition},
for any $H\in G(E)$ and $K,J\in G_{\mathrm{sh}}(H)$,
\[
 d_{G(E)}(K,J)=d_{G_{\mathrm{sh}}(H)}(K,J)
             =d_{G_{\mathrm{lp}}(\phi(H))}(\phi(K),\phi(J)).
\]
In particular, the corresponding components have the same diameter,
number of vertices, and directed cycle structure.
\end{corollary}
\begin{proof}[Proof of Corollary~\ref{cor:scc-distance}]
A path between vertices of an SCC cannot leave that SCC and return.
The first equality follows, and the second follows from the graph
isomorphism of Theorem~\ref{thm:scc-decomposition}.
\end{proof}

For regular equations, the length-preserving SCCs already form a DAG
of linear depth~\cite[Section~3.4]{day_manea}. Theorem~\ref{thm:scc-decomposition}
extends this description to quadratic equations with a linear side:
each SCC is isomorphic to a regular length-preserving component, while
the potential controls the depth of the condensation. The replacement
is chosen separately for each component; it does not assert preservation
of solution sets or an isomorphism between the full Nielsen graphs.

\section{Extension to quadratic equations}\label{sec:extensions}

In this section, we extend the argument for equations with a linear
side to establish NP membership for two subclasses of quadratic word
equations defined by variable occurrence arrangements.
Subsection~\ref{subsec:general-path-shortening} first bounds the length
of Nielsen paths for arbitrary quadratic equations in terms of the
number of $R$-operations. Subsection~\ref{subsec:occurrence-arrangements}
then applies this bound to aligned-block and cyclic-bridge equations.

The occurrence-based subclasses in Subsection~\ref{subsec:occurrence-arrangements}
are not contained in the regular, variable-sparse, or fixed-$k$-ordered
classes~\cite{quadratic_minimal_solutions}.

Throughout this section, $E=(U,V)$ is quadratic, $N=|E|$, and
$m=|\operatorname{Var}(E)|$. A variable is \emph{left-repeated}
or \emph{right-repeated} if it occurs twice on the corresponding side.
A prefix extension $x\mapsto yx$ is an \emph{$R$-operation}, denoted $R(x,y)$,
if $x$ and $y$ are variables repeated on opposite sides immediately
before the operation. The ordered pair records the extended variable
first. A \emph{length-preserving segment} is a maximal segment consisting
of length-preserving edges. A path is called a \emph{directed $R$-nonreuse path}
if no ordered pair $R(x,y)$ occurs twice in any such segment.

\subsection{Extension of path shortening to quadratic word equations}
\label{subsec:general-path-shortening}

We first establish two lemmas used to prove the path-shortening theorem
below. They extend Lemmas~\ref{lem:freezing}
and~\ref{lem:shared-shortening} to quadratic word equations with
repetitions on both sides. Their proofs need only address the
differences from the earlier arguments.

\begin{lemma}
\label{lem:general-freezing}
Let $H$ be any quadratic equation. Let $\phi$ replace every
left-repeated or right-repeated variable of $H$ by a distinct fresh
constant and leave all other symbols unchanged.
For every equation $K$ reachable from $H$ using only shared extensions
and every integer $\ell\ge0$, the following are equivalent:
\begin{enumerate}
\item[(1)] There is a path of shared extensions from $H$ to $K$
of length $\ell$.
\item[(2)] There is a path of length-preserving operations from
$\phi(H)$ to $\phi(K)$ of length $\ell$.
\end{enumerate}
\end{lemma}
\begin{proof}[Proof of Lemma~\ref{lem:general-freezing}]
The proof of Lemma~\ref{lem:freezing} applies unchanged when variables 
repeated on either side are replaced by distinct fresh constants. Indeed, 
shared extensions preserve the number of occurrences of each symbol on 
each side and never extend a variable being replaced. The resulting equation 
is regular, so the same correspondence holds in both directions.
\end{proof}

\begin{lemma}
\label{lem:general-shared-shortening}
For every quadratic equation $H$ and every equation $K$ reachable
from $H$ using only shared extensions, there is a path of shared
extensions from $H$ to $K$ of length at most $C|H|^{11}$,
for a constant $C$ independent of $H$ and $K$.
\end{lemma}
\begin{proof}[Proof of Lemma~\ref{lem:general-shared-shortening}]
The proof of Lemma~\ref{lem:shared-shortening} applies with
Lemma~\ref{lem:general-freezing} in place of Lemma~\ref{lem:freezing}.
The replacement still gives a regular equation of size $|H|$, so
Lemma~\ref{lem:regular-distance} yields the same bound $C_1|H|^{11}$.
Taking $C=C_1$ proves the lemma.
\end{proof}

Using these lemmas, we now prove the following theorem, which bounds
path length in terms of the number of $R$-operations.

\begin{theorem}[Path length controlled by $R$-operations]
\label{thm:R-shortening}
There is a constant $C$ independent of $E$, $F$, and the constant
alphabet such that every finite Nielsen path
from $E$ to an equation $F$, containing $r$ $R$-operations, can be
replaced by a path with the same endpoints of length at most
\[
 C(m+r+N)N^{11}.
\]
\end{theorem}

\begin{proof}[Proof of Theorem~\ref{thm:R-shortening}]
Fix a path from $E$ to $F$ containing $r$ $R$-operations.
If $N=0$, the path is empty and the claim is immediate. Assume $N\ge1$.
For an equation $H$ on the path, let $S_{\mathrm{sh}}(H)$ be its number of shared
variables. We first bound the operations at which we split the path.

Consider a prefix extension $x\mapsto yx$, where $x$ is repeated on
one side. If $y$ is a shared variable, the operation decreases $S_{\mathrm{sh}}$
by one; call it a $P$-operation. If $y$ is repeated on the opposite
side, the operation is an $R$-operation and increases $S_{\mathrm{sh}}$ by one.
If $y$ is a constant or a variable occurring only once, $S_{\mathrm{sh}}$ is unchanged.
Shared extensions leave $S_{\mathrm{sh}}$ unchanged regardless of $y$.
Every length-decreasing operation leaves
$S_{\mathrm{sh}}$ unchanged or decreases it.

Let $p$ be the number of $P$-operations and let $d_{\mathrm{sh}}\ge0$ be the total
decrease in $S_{\mathrm{sh}}$ due to length-decreasing operations. Then
\[
 S_{\mathrm{sh}}(F)=S_{\mathrm{sh}}(E)-p+r-d_{\mathrm{sh}},
 \qquad p+r\le S_{\mathrm{sh}}(E)+2r\le m+2r.
\]
By Lemma~\ref{lem:quadratic-size}, there are at most $N$
length-decreasing operations. Splitting the path at these operations
and at all $P$- and $R$-operations therefore gives at most
$m+2r+N+1$ intervening segments, some of which may be empty.

By the size-change calculation in the proof of Lemma~\ref{lem:potential},
only extensions of variables occurring twice preserve length.
Since all $P$-operations, $R$-operations, and length-decreasing
operations separate segments, the only operations remaining within a
segment are shared extensions and prefix extensions $x\mapsto yx$,
where $x$ occurs twice on one side and $y$ is a constant or a variable
occurring only once.
Once a shared extension has occurred, at least one leading symbol
remains a shared variable throughout the segment. Extending that variable
is a shared extension. Extending the opposite leading variable is
another shared extension, a $P$-operation, or a length-decreasing
operation; the latter two end the segment. Consequently, each segment
consists of the other prefix extensions followed by shared extensions.

The first sequence already has at most $N$ steps. Throughout it,
the same repeated variable $x$ stays at the head of one side.
Let $w$ be the longest prefix of the opposite side consisting of
constants and variables occurring only once. Each operation removes
the first symbol of $w$ and inserts it before the second occurrence
of $x$, which is on the same side as the first. No transferred symbol
returns to the opposite leading prefix, so there are at most
$|w|\le N$ such operations.

The remaining shared-extension sequence can be replaced, by
Lemma~\ref{lem:general-shared-shortening}, by a sequence of length at
most $C_0N^{11}$ with the same endpoints. Each intervening segment
therefore has a replacement of length at most $N+C_0N^{11}$.
As in the proof of Theorem~\ref{thm:main}, the replacements preserve
endpoints and hence concatenate to a path from $E$ to $F$ of length
at most
\[
 (m+2r+N+1)(N+C_0N^{11})+(m+2r+N)
 \le C(m+r+N)N^{11}
\]
for a constant $C$ independent of $E$, $F$, and the constant alphabet.
This proves the required bound.
\end{proof}

Theorem~\ref{thm:R-shortening} reduces the search for polynomial-length
accepting paths to bounding the number of $R$-operations. For any
polynomial-time recognizable class, NP membership follows if every
satisfiable input admits an accepting path with polynomially many
$R$-operations. One sufficient condition is directed $R$-nonreuse:
each length-preserving segment then contains at most one $R$-operation
per ordered variable pair, giving a polynomial bound on their total
number. This yields the following corollary.

\begin{corollary}[Nonreuse certificates]\label{cor:nonreuse-certificate}
Every quadratic word equation admitting an accepting directed
$R$-nonreuse path admits an accepting path of length $O(N^{14})$.
Consequently, satisfiability is in NP for polynomial-time recognizable
classes of quadratic word equations whose satisfiable instances admit
such nonreuse paths.
\end{corollary}
\begin{proof}[Proof of Corollary~\ref{cor:nonreuse-certificate}]
Choose an accepting directed $R$-nonreuse path from $E$.
By Lemma~\ref{lem:quadratic-size}, at most $N$ operations decrease
the size, so the path has at most $N+1$ length-preserving segments.
No operation introduces a new variable. There are $m(m-1)$ ordered
pairs $(x,y)$ of distinct initial variables, and the nonreuse condition
allows each pair at most once in each segment. Thus each segment
contains at most $m(m-1)$ $R$-operations, and their
total number $r$ satisfies
\[
 r\le (N+1)m(m-1)=O(N^3),
\]
since $m\le N$. Theorem~\ref{thm:R-shortening} gives an accepting
path with the same endpoints of length
\[
 O((m+r+N)N^{11})=O(N^{14}).
\]

The NP certificate argument is the same as for
Corollary~\ref{cor:np}, after checking membership in the class.
\end{proof}

\subsection{Classes defined by occurrence arrangements}
\label{subsec:occurrence-arrangements}

We define two classes by the arrangement of variable occurrences and
prove NP membership for each using Corollary~\ref{cor:nonreuse-certificate}.

\begin{definition}[Aligned-block equations]\label{def:aligned-blocks}
A quadratic word equation is an \emph{aligned-block equation} if it
can be written $PS=QT$, where $P,Q,S,T$ are linear and
\[
 \operatorname{Var}(P)\cap\operatorname{Var}(T)=\varnothing,
 \qquad
 \operatorname{Var}(Q)\cap\operatorname{Var}(S)=\varnothing.
\]
We call $P,Q$ the front blocks and $S,T$ the rear blocks.
The aligned-block satisfiability problem asks whether a given
aligned-block equation has a solution.
\end{definition}

\begin{theorem}\label{thm:aligned-np}
The aligned-block satisfiability problem is in NP.
\end{theorem}

We first establish nonreuse and recognition for aligned-block equations.

\begin{lemma}\label{thm:aligned-blocks}
Every Nielsen path from an aligned-block equation is a directed
$R$-nonreuse path.
\end{lemma}
\begin{proof}[Proof of Lemma~\ref{thm:aligned-blocks}]
We first show that the four-block conditions are preserved while both
front blocks $P,Q$ are nonempty. Erasure and cancellation only remove
symbols from the blocks. For a prefix extension $x\mapsto yx$ with
$x$ leading on the left, write $P=xP'$ and $Q=yQ'$, and let $\theta$
be the substitution. The new front blocks are $\theta(P)$ and
$\theta(Q)$ with their leading $y$ removed, and the new rear blocks
are $\theta(S)$ and $\theta(T)$.

We distinguish three cases according to whether and where $x$ occurs
besides the leading occurrence in $P$.

\noindent (i) If $x$ also occurs in $S$, it occurs nowhere on the right.
A variable $y$ in $Q$ cannot occur in $S$, so inserting $y$ before
the second $x$ preserves linearity of $S$. After removing the leading
$y$ from $Q$, its remaining occurrences lie in $P,S$ if it previously
occurred in $P$, in $S,T$ if it previously occurred in $T$, or only
in $S$ if it occurred once. None violates the two disjointness
conditions. A constant $y$ does not affect these conditions.

\noindent (ii) If $x$ also occurs in $Q$, the extension only moves $y$ within
$Q$; $P,S,T$ are unchanged.

\noindent (iii) If $x$ occurs only once, the extension only removes the leading
$y$ from $Q$.

These cases preserve all conditions, and extensions on the right are
symmetric.

For $R(x,y)$ with left-repeated $x$ and right-repeated $y$, case (i)
moves $y$ from $Q$ into $S$, leaving $y$ shared between $S,T$.
$y$ cannot return to a front block: subsequent extensions choose
their prefix symbols from the heads of $P,Q$. Since $y$ occurs only
in $S,T$, it cannot be chosen and inserted into $P$ or $Q$.
Thus it participates in no further $R$-operation while both front
blocks remain nonempty. Once one front block is empty, that side is
linear; Lemma~\ref{lem:right-linear}, or its symmetric version,
precludes any further $R$-operation. The case where $x$ is
right-repeated and $y$ is left-repeated is symmetric.

An unordered pair cannot be reused, since its prefix variable
participates in no later $R$-operation. In particular, the path has
directed $R$-nonreuse.
\end{proof}

The aligned-block condition can be recognized directly from the input.

\begin{lemma}\label{lem:aligned-recognition}
Whether a quadratic word equation of size $N$ is an aligned-block
equation can be decided in $O(N)$ time and space.
\end{lemma}
\begin{proof}[Proof of Lemma~\ref{lem:aligned-recognition}]
The idea is to determine which positions must lie in the rear blocks.
The second occurrence of every same-side repeated variable must do so.
Moreover, if a position lies in a rear block, every position to its
right on that side must also lie in a rear block, as must the other
occurrence of the same variable when it is shared.

Represent these requirements by a directed graph on symbol occurrences:
join successive positions on each side from left to right, and join
the two occurrences of each shared variable in both directions.
Mark first occurrences of same-side repeated variables by $F$ and
second occurrences by $B$. Every vertex reachable from a $B$-marked
vertex must lie in a rear block. Thus reaching an $F$-marked vertex
makes a valid decomposition impossible. Conversely, if no such vertex
is reached, take the reachable positions as rear blocks and all other
positions as front blocks. The rear positions form a suffix on each
side, shared occurrences belong to corresponding blocks, and each
same-side repeated variable has one occurrence in each block.
This gives an aligned decomposition.

To implement this criterion, perform a multi-source search from all
$B$-marked vertices. If the search reaches an $F$-marked vertex, reject;
otherwise, cut each side before its first reachable position, or at
its end if none is reachable. The graph has $O(N)$ vertices and edges,
so its construction and the search take $O(N)$ time and space.
\end{proof}

\begin{proof}[Proof of Theorem~\ref{thm:aligned-np}]
By Lemmas~\ref{lem:acceptance} and~\ref{thm:aligned-blocks}, every
satisfiable aligned-block equation admits an accepting directed
$R$-nonreuse path. Lemma~\ref{lem:aligned-recognition} recognizes the
class in polynomial time. Corollary~\ref{cor:nonreuse-certificate}
therefore gives NP membership.
\end{proof}

We next consider cyclic-bridge equations.

\begin{definition}[Cyclic-bridge equations]\label{def:cyclic-bridge}
A quadratic word equation is a \emph{cyclic-bridge equation} if it has
the form $WYU=YQT$, where $Y$ is a shared variable absent from
$W,U,Q,T$, the words $U,Q,T$ are linear,
\[
 \operatorname{Var}(Q)\cap\operatorname{Var}(U)=\varnothing,
 \qquad
 \operatorname{Var}(W)\cap\operatorname{Var}(UQT)=\varnothing,
\]
and deleting constants and variables occurring only once in $W$
leaves a word $ZZ$, where $Z$ contains each of its variables once.
The cyclic-bridge satisfiability problem asks whether a given
cyclic-bridge equation has a solution.
\end{definition}

\begin{theorem}\label{thm:cyclic-bridge-np}
The cyclic-bridge satisfiability problem is in NP.
\end{theorem}

For cyclic-bridge equations, $R$-free preliminary operations lead to
the aligned-block class or to an equation with a linear side.

\begin{lemma}\label{cor:cyclic-bridge}
Every Nielsen path from a cyclic-bridge equation is a directed
$R$-nonreuse path.
\end{lemma}
\begin{proof}[Proof of Lemma~\ref{cor:cyclic-bridge}]
We show that each operation either preserves the cyclic-bridge form
without using an $R$-operation, or produces an aligned-block equation
or an equation with a linear side. The latter equations admit no
$R$-reuse by Lemma~\ref{thm:aligned-blocks} and preservation of a linear side.

While $W$ is nonempty, the opposite leading variable $Y$ is shared,
so no $R$-operation is possible. We distinguish the following cases
according to the operation performed, treating an empty $W$ separately.

\noindent (i) The shared variable $Y$ is extended, or a leading
variable of $W$ is erased. Write $W=sW'$, where $s$ is its leading
symbol. The extension $Y\mapsto sY$ gives $W'sYU=YQT$.
Let $w$ be the word obtained from $W$ by deleting constants and
single-occurrence variables. If $s$ is a constant or a single-occurrence
variable, $w$ is unchanged. Otherwise, $w=sDsD$, which becomes $DsDs$ after rotation and
still has the required form.
Erasing a repeated variable from $W$ removes it from each occurrence
of $Z$ in $w=ZZ$, preserving this form; erasing a single-occurrence
variable leaves $w$ unchanged. Neither operation changes
$U,Q,T$ or introduces shared variables between $W$ and these words.
Thus both operations preserve the cyclic-bridge conditions.

\noindent (ii) A repeated leading variable $x$ of $W$ is extended
by $Y$. Write $W=xAxB$. After deleting constants and single-occurrence
variables, $W$ has the form $xDxD$, so every other repeated variable
occurs once in $A$ and once in $B$. Hence $A,B$ are linear.
The resulting equation is
\[
 (xAY)(xBYU)=QT.
\]
The displayed left blocks and the right blocks $Q,T$ satisfy
Definition~\ref{def:aligned-blocks}.

\noindent (iii) A single-occurrence leading variable of $W$ is
extended by $Y$, or $Y$ is erased. The resulting equations are
$WYU=QT$ and $WU=QT$, respectively. Let $w=ZZ$ be the word obtained from $W$ by deleting constants and
single-occurrence variables. If $Z$ is nonempty, cutting $W$ just
before the first variable of the second $Z$ gives two linear words;
if $Z$ is empty, $W$ itself is linear. Append $YU$ or $U$
to the rear block. These decompositions satisfy
Definition~\ref{def:aligned-blocks}. As in (ii), the assumptions ensure the variable-disjointness
conditions in Definition~\ref{def:aligned-blocks}.

\noindent (iv) $W$ is empty. The equation is $YU=YQT$, whose left
side is linear, so no later $R$-operation is possible.

These cases cover all operations: when $W$ is nonempty, its leading
symbol is distinct from $Y$, and each leading variable can be erased
or extended. Before reaching an aligned-block equation or an equation
with a linear side, the path uses no $R$-operations. Afterwards,
Lemma~\ref{thm:aligned-blocks} or preservation of a linear side gives
directed $R$-nonreuse.
\end{proof}

\begin{proof}[Proof of Theorem~\ref{thm:cyclic-bridge-np}]
By Lemmas~\ref{lem:acceptance} and~\ref{cor:cyclic-bridge}, every
satisfiable cyclic-bridge equation admits an accepting directed
$R$-nonreuse path. To recognize the class, the leading shared variable
$Y$ on the right determines $W,U$; enumerate the cuts $QT$ and check the conditions in
Definition~\ref{def:cyclic-bridge}. This takes polynomial time.
Corollary~\ref{cor:nonreuse-certificate} therefore gives NP membership.
\end{proof}

\paragraph{Comparison with known subclasses.}
To compare these classes with the known subclasses, let
$A=x_1\cdots x_t$ and $B=z_1\cdots z_t$ consist of distinct variables.
The equations $AA=BB$ and $AAy=yBB$, with a fresh variable $y$,
satisfy the aligned-block and cyclic-bridge conditions, respectively.
Neither has a linear side. For sufficiently large $t$, their numbers
of repeated variables exceed the variable-sparse bound. In any
$k$-ordered decomposition, all variables in $A$ must belong to one
block, since their occurrences span the two copies; hence $k\ge t$.
Thus neither class is contained in the regular, variable-sparse, or
fixed-$k$-ordered classes~\cite{quadratic_minimal_solutions}.

\section{Conclusion and future work}\label{sec:scope}

We have shown that quadratic word equations with a linear side admit
polynomial-length Nielsen paths to every reachable equation. Together
with the characterization of satisfiability by acceptance and the known
NP-hardness for regular equations, this establishes NP-completeness for
this class.

For equations with a linear side, each SCC is isomorphic to the
length-preserving reachability graph of a regular equation, and the
condensation graph has depth at most $|U|+2|V|$. This structure combines
short paths within components with only linearly many transitions
between them.

For general quadratic word equations, we bound path length in terms
of the number of $R$-operations. This bound establishes NP membership
for subclasses defined by variable occurrence arrangements.

A natural next step is to determine whether every satisfiable quadratic
word equation admits an accepting path with polynomially many
$R$-operations. By our path bound, a positive answer would establish
NP membership for the general quadratic case.

\section*{AI Usage Disclosure}
The author used OpenAI Codex throughout the research and preparation of
this manuscript, including iterative discussions of mathematical ideas
and proofs, computational experiments, code development, literature
searches, and manuscript drafting and editing. Its involvement cannot
be meaningfully isolated to particular results or passages. The author
also used Anthropic Claude for computational experiments.
All mathematical claims, proofs, and references were independently
checked by the author, who takes full responsibility for the content
of this paper.

\bibliographystyle{plainurl}
\bibliography{references}
\end{document}